\documentclass[
10pt
]{article}

\usepackage[%
left=0.8in,
right=0.8in,
bottom=0.8in,
top=1in,
footskip=0.2in,
]{geometry}

\usepackage{color}
\definecolor{LinkColor}{rgb}{0,0,1}
\definecolor{lbcolor}{rgb}{0.85,0.85,0.85}
\definecolor{FrameColor}{rgb}{0.85,0.85,0.85}

\usepackage[ngerman, french, english]{babel}
\usepackage[utf8]{inputenc}
\usepackage{enumerate}
\usepackage[normalem]{ulem}

\usepackage[babel,french=guillemets,german=swiss]{csquotes}

\usepackage{graphicx}
\usepackage{caption}
\usepackage{amsmath}
\usepackage{amssymb}
\usepackage{dsfont}
\usepackage{nicefrac}
\usepackage{empheq}
\allowdisplaybreaks
\usepackage{upgreek}

\usepackage[%
pdftitle={Titel},
pdfauthor={Autor},
pdfcreator={LaTeX, LaTeX with hyperref and KOMA-Script},
pdfsubject={Betreff}, 
pdfkeywords={Keywords}
]{hyperref} 

\hypersetup{colorlinks=true,
	linkcolor=LinkColor,%
	anchorcolor=LinkColor,%
	citecolor=LinkColor,%
	filecolor=LinkColor,%
	menucolor=LinkColor,%
	urlcolor=LinkColor,%
	allcolors=LinkColor}

\usepackage[numbers]{natbib}
\usepackage{amsthm}

\newtheoremstyle{tstyle}
{15pt}	
{5pt}	
{\itshape}	
{}	
{\bfseries}	
{.}	
{0.5em}	
{}	

\theoremstyle{tstyle}

\newtheorem{theorem}{Theorem}

\newtheorem{proposition}{Proposition}

\newtheorem{definition}{Definition}

\newtheoremstyle{cstyle}
{15pt}	
{5pt}	
{}	
{}	
{\bfseries}	
{}	
{0.2222em}	
{}	

\theoremstyle{cstyle}

\makeatletter
\g@addto@macro{\thm@space@setup}{\thm@headpunct{}}
\renewenvironment{proof}[1][\proofname]{\par
	\pushQED{\qed}%
	\normalfont \topsep6\p@\@plus6\p@\relax
	\trivlist
	\item[\hskip\labelsep
	\bfseries
	#1\@addpunct{\,}]\ignorespaces
}{%
	\popQED\endtrivlist\@endpefalse
}
\makeatother

\makeatletter
\g@addto@macro{\thm@space@setup}{\thm@headpunct{}}
\newenvironment{sketch-proof}[1][Sketch of the proof]{\par
	\pushQED{\qed}%
	\normalfont \topsep6\p@\@plus6\p@\relax
	\trivlist
	\item[\hskip\labelsep
	\bfseries
	#1\@addpunct{\,}]\ignorespaces
}{%
	\popQED\endtrivlist\@endpefalse
}
\makeatother

\def\RR{\mathds R}
\def\NN{\mathds N}

\newcommand\E{\mathcal{E}}
\renewcommand\L{\mathcal{L}}
\renewcommand\P{\mathcal{P}}

\def\eps{\varepsilon}
\def\supp{\textnormal{supp\,}}

\def\BR{{B_R(0)}}

\def\ddt{\frac{\mathrm d}{\mathrm dt}}
\def\ddr{\frac{\mathrm d}{\mathrm dr}}

\def\dtau{\;\mathrm d\tau}

\def\dx{\;\mathrm dx}
\def\dy{\;\mathrm dy}
\def\dv{\;\mathrm dv}

\def\ds{\;\mathrm ds}

\def\dxv{\;\mathrm d(x,v)}

\def\delxi{\partial_{x_i}}
\def\delt{\partial_{t}}
\def\delx{\partial_{x}}

\def\grad{\nabla}
\def\laplace{\Delta}

\def\Mf{\mathbf{f}\hspace{1pt}}

\def\tand{\quad\text{and}\quad}
\def\twith{\quad\text{with}\quad}

\def\a{\alpha}
\def\fini{\mathring f}
\def\suma{\sum_{\a=\pm 1}}
\def\sf{\raisebox{-1pt}{\large{\textnormal{f}}}}

\def\itema{\item[\textnormal{(a)}]}
\def\itemb{\item[\textnormal{(b)}]}
\def\itemc{\item[\textnormal{(c)}]}

\def\itemi{\item[\textnormal{(i)}]}
\def\itemii{\item[\textnormal{(ii)}]}
\def\itemiii{\item[\textnormal{(iii)}]}
\def\itemiv{\item[\textnormal{(iv)}]}

\newcommand{\Underset}[3][0pt]{\ensuremath{\underset{\raise#1\hbox{\small\ensuremath{#2}}}{#3}}}
\newcommand{\Overset}[3][0pt]{\ensuremath{\overset{\raise#1\hbox{\small\ensuremath{#2}}}{#3}}}

\newcommand{\comma}{\,{,}\,}

\begin{document}
	
\begin{center}	
	\LARGE{\bfseries Confined steady states\\ of a Vlasov-Poisson plasma\\ in an infinitely long cylinder}\\[8mm]
	\normalsize{Patrik Knopf}\\[2mm]
	\textit{University of Regensburg, 93040 Regensburg, Bavaria, Germany}\\[2mm]
	\texttt{Patrik.Knopf@mathematik.uni-regensburg.de}\\[-3mm]
	
	\begin{minipage}[h]{0.275\textwidth}
	\begin{flushright}
		\vspace{-2pt}
		\includegraphics[scale=0.05]{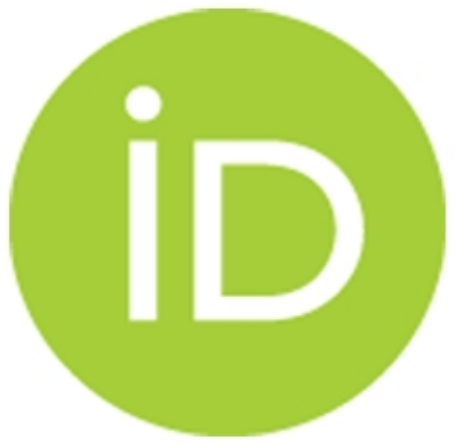} 
	\end{flushright}
	\end{minipage}
	\begin{minipage}[h]{0.5\textwidth}
		\hspace{-12pt}
		\href{https://orcid.org/0000-0003-4115-4885}{orcid.org/0000-0003-4115-4885}
	\end{minipage}

	\bigskip
	\color{white}
	\textit{Please cite as:} \\  Math. Meth. Appl. Sci. (2019) \\ URL
	\color{black}
	\bigskip
\end{center}

\vfill

\abstract{We consider the two-dimensional Vlasov-Poisson system to model a two-component plasma whose distribution function is constant with respect to the third space dimension. First, we will show how this two-dimensional Vlasov-Poisson system can be derived from the full three-dimensional system. The existence of compactly supported steady states with vanishing electric potential in a three-dimensional setting has already been investigated by A.\,L. Skubachevskii  \cite{skubachevskii}. We will show that his approach can easily be adapted to the two-dimensional system. However, our main result is to prove the existence of compactly supported steady states even with a \textit{nontrivial} self-consistent electric potential.\\
\phantom{}
\phantom{}
}

\vfill




\section{Introduction}

The investigation of a high-temperature plasma under the influence of an exterior magnetic field is essential in fusion research. One of the most promising approaches to generate thermonuclear fusion power is magnetic confinement. The idea is to use magnetic fields to control the plasma in such a way that it keeps a certain distance to the reactor wall. This is necessary because a plasma that collides with the wall will cool down (which makes the fusion process impossible) and will possibly damage the reactor. 
The most common reactor types for magnetic confinement are toroidal devices, such as Tokamaks, Stellerators and reversed field pinches, as well as linear confinement devices, for example z-pinch configurations.\\
\begin{figure}[h!]
	\centering
	\textbf{Rough sketch of a cylindrical device for plasma confinement}\\[2mm]
	\includegraphics[width=0.8\textwidth]{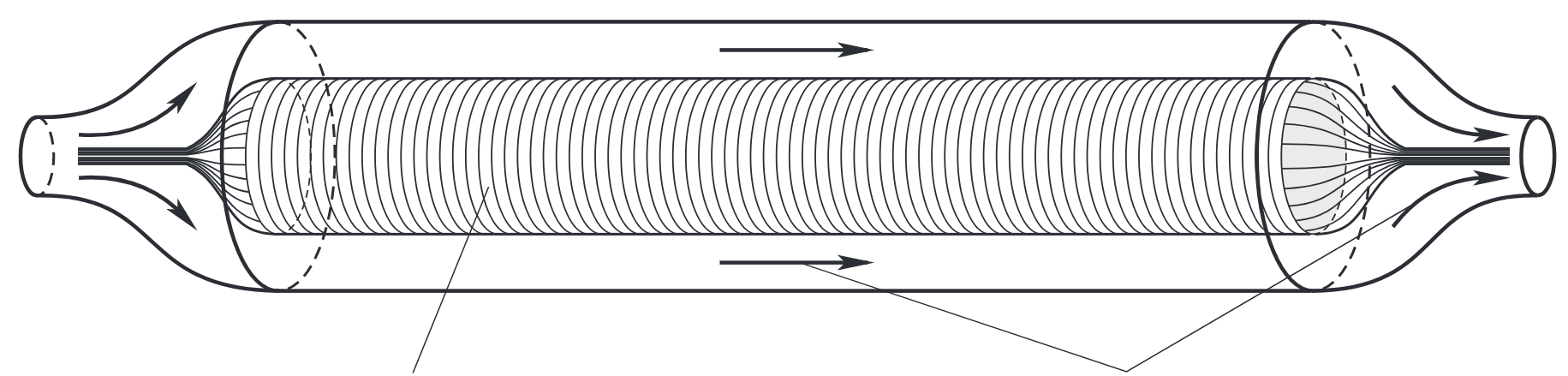}\\[-2mm]
	\textbf{plasma} \hspace{50mm} \textbf{magnetic field}\\ 
	\caption{ (image adapted from Fig.\,2 by A.\,L. Skubachevskii  \cite{skubachevskii}, slightly modified) }
	\label{fig1}
\end{figure}

\newpage
\noindent The results presented in this paper are motivated by such confinement problems. We discuss a simplified mathematical model to describe the behavior of a two-component plasma in a long cylindrical reactor (cf. \autoref{fig1} for a rough illustration). In particular, we intend to prove existence of steady states to the Vlasov-Poisson system that are \textit{confined} inside the reactor if the exterior magnetic field is sufficiently strong. This means that the whole plasma is supposed to be contained in a cylinder with a given radius $R>0$ that is smaller than the radius of the reaction chamber. As the cylindrical device is assumed to be very long, we will suppose (for mathematical reasons) that its length is infinite.\\[1ex]
The motion of a plasma can be described by the kinetic equation\vspace{-1mm}
\begin{align}
\label{BE}
&\;\partial_t \Mf^\a + v\cdot \partial_x \Mf^\a + \frac{\a}{m_\a}F \cdot \partial_v \Mf^\alpha = \left(\frac{\delta f}{\delta t}\right)_\text{coll}, \quad \a=\pm 1
\end{align}
which is called the \textit{Boltzmann equation} \cite{miyamoto}. 
Here $\Mf^\a = \Mf^\a(t,x,v)\ge 0$ (with $t\ge 0, x,v\in\RR^3$) denotes the (scalar) distribution function of positively charged ions ($\alpha=+1$) or of electrons ($\alpha=-1$). For any measurable set $\mathcal M\subset \RR^6 = \RR^3\times\RR^3$, the term\vspace{-1mm}
\begin{align*}
\sum_{\a=\pm 1} \a \int\limits_{\mathcal M}  \Mf^\a(t,x,v)\dxv 
\end{align*}
represents the total charge of all particles having space coordinates $x\in\RR^3$ and velocity coordinates $v\in\RR^3$ with $(x,v)\in\mathcal M$ at time $t\ge 0$.
If $\alpha$ appears in the index, we will just write $+$ or $-$ instead of $+1$ or $-1$. We assume that each positively charged ion (cation) has the charge $e>0$ and the mass $m_+>0$ while the charge and mass of the electrons are given by  $-e<0$ and $m_->0$.\\[1ex]
The term $(\delta f/\delta t)_\text{coll}$ describes effects due to collisions of the particles. If the plasma has a very high temperature and is sufficiently rarefied, one can argue that this collision term may be neglected \cite{miyamoto} (especially if short-time processes are described  \cite{kaufmann,klimontovich}). The collisionless Boltzmann equation is usually referred to as the \textit{Vlasov equation} (which was originally proposed by A.\,A.~Vlasov \cite{vlasov}). \\[1ex]
Moreover, $F=F(t,x)$ denotes a self-consistent force field that is generated by the charged particles via electromagnetic induction. Vice versa, $F$ causes the movement of the particles. It can best be described by Maxwell's equations which leads to the \textit{Vlasov-Maxwell system}. However, in mathematical physics a simplified model where magnetic induction and electrodynamic effects are neglected is also frequently investigated. This means that only electrostatic interactions of the particles are considered and thus, $F=e E$ where $E$ is an electric field described by Coulomb's law. This leads to the so-called  \textit{Vlasov-Poisson system} which has already been studied in various papers in the mathematical literature \cite{kurth,batt,pfaffelmoser,lions-perthame,schaeffer,rein}.\\[1ex]
Furthermore, we add an exterior magnetic field $B=B(t,x)$ to the Vlasov equation which influences the particles via Lorentz force $\tfrac{e}{c} (v\times B)$. Here, the constant $c$ denotes the speed of light. Since magnetic induction caused by the particles is neglected, we also suppose that the magnetic field is not influenced by the plasma dynamics. The overall system of equations reads as follows:
\begin{empheq}[left=\empheqlbrace]{align}
\label{VE}
&\;\partial_t \Mf^\a + v\cdot \partial_x \Mf^\a + \frac{\a e}{m_\a}\left[E + \frac{1}{c}(v\times B) \right]\cdot \partial_v \Mf^\alpha = 0, \quad \a=\pm 1
\\[0.15cm]
\label{PE}
&\;E = -\delx\psi,\quad -\laplace_x \psi = 4\pi e \sum_{\a=\pm 1}\a\,\varrho^\a,
\quad \underset{|x|\to\infty}{\lim}\psi(t,x) = 0 \\
\label{RE}
&\;\varrho^\a(t,x) = \int \Mf^\a (t,x,v) \dv
\end{empheq}
Here, the functions $\varrho^\a$ denote the spatial charge densities associated with $f^\a$.
The scalar function $\psi=\psi(t,x)$ stands for the self-consistent electrostatic potential induced by the charge of the ions and electrons. The corresponding electric field $E=E(t,x)$ is given by $E=-\delx\psi$.\\[1ex]
\noindent Assuming that $\Mf^\a$ is locally integrable, we can solve Poisson's equation explicitly (at least in the sense of distributions) by the Newtonian potential. We obtain
\begin{align*}
\psi(t,x)= \int \frac{1}{|x-y|} \suma \a\varrho^\a(t,y) \dy 
= \iint \frac{1}{|x-y|} \sum_{\alpha=\pm 1} \a\Mf^\a(t,y,w) \;\mathrm dw\mathrm dy
\end{align*}
and hence the Vlasov-Poisson equation can be expressed alternatively by
\begin{empheq}[left=\empheqlbrace]{align}
\label{VEV2}
&\;\partial_t \Mf^\a + v\cdot \partial_x \Mf^\a + \frac{\a e}{m_\a}\left[E + \frac{1}{c}(v\times B) \right]\cdot \partial_v \Mf^\alpha = 0, \quad \a=\pm 1 
\\[0.25cm]
\label{CE}
&\;E(t,x) = -e\iint \frac{x-y}{|x-y|^3} \sum_{\alpha=\pm 1} \a\,\Mf^\a(t,y,w) \;\mathrm dw\mathrm dy. 
\end{empheq}
By imposing the condition
\begin{align}
\label{IC}
\Mf^\a\big\vert_{t=0} = {\mathring\Mf}\phantom{}^\a, \quad \a=\pm 1
\end{align}
for initial distribution functions ${\mathring \Mf}\phantom{}^\a \in C^1_c(\RR^6)$ we obtain an initial value problem. In the absence of an exterior magnetic field ($B= 0$) and for only one sort of particles (without loss of generality $\Mf^+= 0$) a first local existence and uniqueness result to this initial value problem was proved by R. Kurth  \cite{kurth}. Later J. Batt  \cite{batt} established a continuation criterion which claims that a local solution can be extended as long as its velocity support is under control. Finally, two different proofs for global existence of classical solutions were established independently and almost simultaneously, one by K. Pfaffelmoser  \cite{pfaffelmoser} and one by P.-L. Lions and B. Perthame  \cite{lions-perthame}. Later, a greatly simplified version of Pfaffelmoser’s proof was published by J. Schaeffer  \cite{schaeffer}. This means that the following result is established: Any nonnegative initial datum $\mathring \Mf\phantom{}^- \in C^1_c(\RR^6)$ launches a global classical solution $\Mf^- \in C^1([0,\infty[\times\RR^6)$ of the initial value problem (\ref{VE}-\ref{RE},\ref{IC}) (or (\ref{VEV2}-\ref{IC}) respectively). Moreover, for any time $t\in[0,\infty[$, $\Mf^-(t)=\Mf^-(t,\cdot,\cdot)$ is compactly supported in $\RR^6$. For more information we recommend to consider the review article about the Vlasov-Poisson system by G.~Rein  \cite{rein}. For general fields $B\in C([0,T];C^1_b(\RR^3;\RR^3))$ the problem has already been investigated \cite{knopf,knopf-weber}. The existence and uniqueness result for global classical solutions holds true in this case as the Pfaffelmoser-Schaeffer proof can be adapted to this problem. A proof of this assertion was previously given by the author \cite{knopf}. This proof can easily be modified to prove the analogous result for a plasma with two sorts of particles (i.e., $\Mf^+$ is nonzero). \\[1ex]
As stated above, we want to study a plasma in an infinitely long cylinder. Without loss of generality the cylinder axis is supposed to point in the $\vec e_3$-direction. We assume that the distribution function of the plasma does not depend on $x_3$, i.e., it is identical in any cross section parallel to the $(x_1,x_2)$-plane. In particular, $\Mf^\a$ is supposed to have the following shape:
\begin{align*} 
\Mf^\a(t,x,v) = f^\a(t,x_1,x_2,v_1,v_2)\, \chi^\a(v_3), \quad \alpha=\pm 1,\quad (t,x,v)\in[0,T[\times \RR^4
\end{align*}
where $f^\a:[0,T[\times\RR^4\to [0,\infty[$ is continuously differentiable with respect to all its variables, $\supp\,f^\a(t)$ is compact for all $t\in[0,T[$ and
\begin{align*}
\chi^\a \in C^1_c(\RR) 
\twith
\int\limits_{-\infty}^\infty \chi^\a(s) \ds = 1.
\end{align*}
This kind of symmetry can be preserved only if the magnetic field $B$ is also independent of $x_3$. We assume that it points only in the $\vec e_3$-direction, i.e., $B=b\,\vec e_3$ where $b=b(t,x_1,x_2)$ is a scalar function. Then
\begin{align}
\label{B2D}
v\times B = b\, (v_1,-v_2,0)^T.
\end{align}
The electric field can also be simplified subject to this symmetry assumption. From \eqref{CE} we obtain that
\begin{align*}
E_i(t,x_1,x_2) 
&= -e \iiint \frac{y_i}{|y|^3} \;dy_3 \sum_{\alpha=\pm 1} \a\,\rho^\a(t,x_1-y_1,x_2-y_2)\;  \mathrm dy_2\mathrm dy_1
\end{align*}
where
\begin{align}
\rho^\a(t,x_1,x_2) := \int \Mf^\alpha(t,x_1,x_2,0,v)\; \mathrm dv = \iint f^\alpha(t,x_1,x_2,v_1,v_2)\; \mathrm dv_2\mathrm dv_1
\end{align}
denotes the charge density of the particles in the $(x_1,x_2)$-plane (or any plane that is parallel to the $(x_1,x_2)$-plane respectively). We can easily conclude that $E_3 = 0$ as the function $y_3 \mapsto y_3\,|y|^{-3}$ is odd. For $i\in\{1,2\}$ we can compute the $x_3$-integral explicitly and obtain that
\begin{align}
\label{E2D}
\begin{aligned}
E_i(t,x_1,x_2) 
&= -2e \iint \frac{y_i}{y_1^2+y_2^2} \sum_{\alpha=\pm 1} \a\,\rho^\a(t,x_1-y_1,x_2-y_2)\;  \mathrm dy_2\mathrm dy_1\\
&= -2e \iint \frac{x_i-y_i}{(x_1-y_1)^2+(x_2-y_2)^2} \sum_{\alpha=\pm 1} \a\,\rho^\a(t,y_1,y_2)\;  \mathrm dy_2\mathrm dy_1.
\end{aligned}
\end{align}
Now, using \eqref{B2D} and \eqref{E2D} and assuming $f^\a$ to be sufficiently regular, we have
\begin{align*}
\left\{\delt f^\a + \begin{pmatrix}v_1\\v_2\end{pmatrix}\cdot\begin{pmatrix}\partial_{x_1} f^\alpha \\ \partial_{x_2} f^\alpha \end{pmatrix} 
+ \frac{\a e}{m_\a}\left[ \begin{pmatrix} E_1\\E_2 \end{pmatrix} + \frac b c \begin{pmatrix}v_2\\-v_1\end{pmatrix} \right]\cdot \begin{pmatrix}\partial_{v_1} f^\alpha \\ \partial_{v_2} f^\alpha \end{pmatrix}  \right\} \chi^\alpha(v_3) = 0 
\end{align*}
on $\RR^2$ for any $v_3\in\RR$. This leads to the two-dimensional Vlasov-Poisson system
\begin{empheq}[left=\empheqlbrace]{align}
&\;\partial_t f^\a + v\cdot \partial_x f^\a + \frac{\a e}{m_\a}\left[E + \frac{b}{c}v^\bot \right]\cdot \partial_v f^\alpha = 0, \quad \a=\pm 1
\\[0.25cm]
\label{CE2}
&\;E(t,x) = -2e \int \frac{x-y}{|x-y|^2} \sum_{\alpha=\pm 1} \a\,\rho^\a(t,y)\;  \mathrm dy.
\end{empheq}
where $x,y$ and $v$ now denote vectors in $\RR^2$, $v^\bot := (v_2,-v_1)^T$ and ${E:=(E_1,E_2)^T}$. Note that the electric field can also be expressed as a gradient field  $E=-\delx\phi$ with the potential
\begin{align}
\phi(t,x) = -2e \iint \ln(|x-y|)\,\sum_{\alpha=\pm 1} \a\,f^\a(t,y,w)\;  \mathrm dy\mathrm dw.
\end{align}
Then, analogous to the three-dimensional case, $\phi$ is a solution of the two-dimensional Poisson equation
\begin{align}
-\laplace \phi = 4\pi e \,\sum_{\alpha=\pm 1} \a\,\rho^\a.
\end{align}
However, $\phi$ does not satisfy the homogeneous boundary condition from \eqref{PE} but we can find an alternative condition: Recall that $f^\a$ was supposed to be continuously differentiable and $f^\a(t)$ was assumed to be compactly supported for all $t\in [0,T[$. We define the quantity
\begin{align}
M:=\sum_{\a =\pm 1}\a \int \rho^\a(t,x) \dx =\sum_{\a =\pm 1}\a \iint f^\a(t,x,v) \dxv.
\end{align}
By differentiating under the integral followed by integration by parts one can easily show that $\ddt M = 0$, i.e.,$M$ does not depend on time. Now we fix some arbitrary time $t\in [0,T[$. As it was assumed that $f^\a(t)$ is compactly supported in $\RR^4$ there exists $R>0$ such that $\supp f^\alpha(t) \subset B_R(0)$. Let now $x\in\RR^4$ with $|x|\ge 2R$ be arbitrary. Applying the mean value theorem on $\ln(|x|)-\ln(|x-y|)$ we obtain that
\begin{align*}
\big| \phi(t,x) + 2Me\,\ln(|x|) \big| \le \frac{4Me}{|x|} \to 0, \quad |x|\to\infty.
\end{align*}
Hence, we can conclude that $\phi$ is the unique solution of the boundary value problem
\begin{align}
-\laplace \phi = 4\pi e \,\sum_{\alpha=\pm 1} \a\,\rho^\a, \quad \underset{|x|\to\infty}{\lim} \phi(t,x) + 2Me\,\ln(|x|) = 0.
\end{align}
This means that the two-dimensional Vlasov-Poisson system can be formulated alternatively by
\begin{empheq}[left=\empheqlbrace]{align}
\label{VE2}
&\;\partial_t f^\a + v\cdot \partial_x f^\a + \frac{\a e}{m_\a}\left[E + \frac{b}{c}v^\bot \right]\cdot \partial_v f^\alpha = 0, \quad \a=\pm 1 
\\[0.25cm]
\label{PE2}
&\;E = -\delx\phi,\;\; -\laplace \phi = 4\pi e\sum_{\alpha=\pm 1} \a\,\rho^\a, \;\; \underset{|x|\to\infty}{\lim} \phi(t,x) + 2Me\,\ln(|x|) = 0, \\
\label{RE2}
&\;\rho^\a(t,x) = \int f^\a (t,x,v) \dv
\end{empheq}
The adjusted initial condition now reads as follows:
\begin{align}
\label{IC2}
f^\a\big\vert_{t=0} = \fini^\a, \quad a=\pm 1
\end{align}
where $f^\a \in C^1_c(\RR^4),\, \a=\pm 1$ are given initial distributions. Existence and uniqueness of global classical solutions of the initial value problem (\ref{VE2}-\ref{IC2}) in the case $b= 0$ has been established by S. Ukai and T.~Okabe \cite{ukai-okabe}. A similar result for $b= 0$ and only one sort of particles (without loss of generality $f^+=0$) is presented by S. Wollman  \cite{wollman}. Moreover, Wollman showed that the support of the solution is compact and bounded in the following sense: For any $T>0$ there exists $R>0$ such that the classical solution of the initial value problem satisfies the condition
\begin{align*}
\forall t\in[0,T]: \quad \supp f^-(t) \subset \BR.
\end{align*}
This result can easily be generalized to a plasma with two sorts of particles. Certainly, these global existence and uniqueness results remain correct if a nontrivial exterior magnetic field given by $b\in C\big([0,T];C^1_b(\RR^2)\big)$ is added. 
However, in this paper, we will not consider time-dependent but stationary solutions of the two-dimensional Vlasov-Poisson system which are also referred to as steady states. This means that we are looking for solutions of the stationary Vlasov-Poisson system:\\
\begin{empheq}[left=\empheqlbrace]{align}
\label{SVE}
&\;v\cdot \partial_x f^\a + \frac{\a e}{m_\a}\left[E + \frac{b}{c}v^\bot \right]\cdot \partial_v f^\alpha = 0, \quad \a=\pm 1 
\\[0.25cm]
\label{SPE}
&\;E = -\delx\phi,\;\; -\laplace \phi = 4\pi e\sum_{\alpha=\pm 1} \a\,\rho^\a, \;\; \underset{|x|\to\infty}{\lim} \phi(x) + 2Me\,\ln(|x|) = 0, 
\\
\label{SRE}
&\;\rho^\a(x) = \int f^\a (x,v) \dv, \quad \a=\pm 1
\end{empheq}

\bigskip\noindent To be precise, a steady state of the Vlasov-Poisson system is defined as follows:
\pagebreak[3]
\begin{definition}
	$\;$
	\begin{enumerate}
		\itema A pair of functions $f^\a: \RR^2 \times \RR^2 \to [0,\infty[ ,\;\a=\pm 1$ is called a steady state of the Vlasov-Poisson system if the following holds:
		\begin{enumerate}
			\itemi The functions $f^\a$ are continuously differentiable.
			\itemii The induced spatial densities $\rho^\a$ and the electric field $E$ are continuously differentiable. The potential $\phi$ is twice continuously differentiable.
			\itemiii The functions $f^-,f^+,\rho^-,\rho^+,\phi$ and $E$ satisify the stationary Vlasov-Poisson system \textnormal{(\ref{SVE},\ref{SPE},\ref{SRE})} on $\RR^2\times\RR^2$.
		\end{enumerate}
		\itemb We say that a steady state $f^\a$ is compactly supported with radius $R>0$ if for both $\a=\pm1$,
		\begin{align*}
		f^\a(x,v)=0 \quad \text{if}\quad |v|>R \;\;\text{or}\;\; |x|>R.
		\end{align*}
		\itemc We say that a steady state $f^\a$ of the Vlasov-Poisson system has finite charge if for both $\a=\pm 1$, $$\|f^\a\|_{L^1(\RR^4)}<\infty.$$ 
	\end{enumerate}
\end{definition}
\noindent Of course steady states of various Vlasov systems have already been investigated in the literature. Existence of stationary solutions in bounded domains has been established by G.~Rein \cite{rein2} for the Vlasov-Poisson system and by J. Batt and K. Fabian \cite{batt2} for the relativistic Vlasov-Maxwell system. Steady states of the nonrelativistic Vlasov-Maxwell system in one space and two velocity dimensions are studied by Y.~Guo and C.~Grotta~Ragazzo \cite{guo}. Stationary solutions of the Vlasov-Poisson-Boltzmann system are discussed by R.~Duan, T. Yang and C. Zhu \cite{duan} and steady states of the Vlasov-Poisson-Fokker-Planck system are investigated by K. Dressler \cite{dressler} and also R.\,T. Glassey and J. Schaeffer \cite{glassey}. There are also many papers on steady states of the gravitational Vlasov-Poisson system that is used to describe stellar dynamics. We will not itemize all of them but only refer to a paper by T.~Ramming and G.~Rein \cite{ramming} where many important papers concerning this topic are listed in the references section.


\bigskip

\section{Steady states with vanishing electric potential}

In this section we will investigate steady states with vanishing electric potential, i.e., we assume that $\phi=0$. This directly implies that $E=0$ and from Poisson's equation $-\laplace\phi =4\pi e (\rho^+ - \rho^-)$ we obtain that $\rho^- = \rho^+$ which also means $M=0$. Hence the stationary Vlasov-Poisson system with vanishing electric potential reads as follows:
\begin{empheq}[left=\empheqlbrace]{align}
\label{SVEV}
&\;v\cdot \partial_x f^\a + \frac{\a e}{cm_\a}\, b\, v^\bot \cdot \partial_v f^\alpha = 0, \quad \a=\pm 1 
\\
\label{SPEV}
& \;\rho^- = \rho^+ \\
\label{SREV}
&\;\rho^\a(x) = \int f^\a (x,v) \dv , \quad \a = \pm 1
\end{empheq}
For the three-dimensional Vlasov-Poisson system, steady states of that type have already been investigated by A.\,L.\,Skubachevskii  \cite{skubachevskii}. Here, for the two-dimensional system, only slight modifications to his approach are necessary. 	We assume that $b\equiv \beta$ for some constant $\beta>0$. Then 
\begin{gather*}
\E^\a(x,v):=\tfrac 1 2 m_\a |v|^2,\\[2mm]
\L^\a(x,v):= \tfrac 1 2\big|x+\a\tfrac{cm_\a}{e\beta}v^\bot\big|^2,\\[2mm]
\P^\a(x,v):= \L^\a(x,v) - \tfrac{c^2 m_\a}{e^2 \beta^2} \E^\a(x,v)=\tfrac 1 2 |x|^2 + \a \tfrac{cm_\a}{e\beta}x\cdot v^\bot 
\end{gather*}
are solutions of the stationary Vlasov equation \eqref{SVEV} as they are constant along solutions of the characteristic system
\begin{align*}
\dot x = v,\quad \dot v = \tfrac{\a e}{cm_\a}\, \beta\, v^\bot.
\end{align*}
In a physical sense, $\E^\a(x,v)$ represents the kinetic energy of a particle at the point $x$ with velocity $v$. Because of the vanishing electric potential it is also the total energy. The quantity $m_\a x\cdot v^\bot$ denotes the angular momentum of such a particle. The idea behind this is to construct steady states that depend only on these conserved quantities. The following theorem presents three different types steady states where type (c) is especially interesting as it is compactly supported and has finite charge.
\begin{theorem}
	$\,$
	\begin{itemize}
		\itema Let $\sigma_1^\a \in C^1([0,\infty[),\; \a=\pm 1$ be nonnegative functions which satisfy
		\begin{align}
		\forall \eta\ge 0:\; \sigma_1^-(\eta\, m_-) = \sigma_1^+(\eta\, m_+)
		\quad\tand\quad
		\sigma_1^\a(\eta) = 0, \;\;\a=\pm 1, \;\;\text{if}\;\; \eta>E_0 
		\end{align}
		for some constant $E_0>0$. Then $f_1^\a:=\sigma_1^\a\circ \E$ is a steady state of the Vlasov-Poisson system.
		\itemb Let $\sigma_2^\a \in C^1([0,\infty[),\;\a=\pm 1$ be nonnegative functions which satisfy
		\begin{align}
		\label{COND1}
		\begin{gathered}
		\forall \lambda\in\RR: \;\frac{\sigma_2^-(\lambda)}{m_-^2}
		= \frac{\sigma_2^+(\lambda)}{m_+^2} \tand
		\sigma^\a_2(\lambda)=0, \;\;\a=\pm 1, \;\;\text{if}\;\; \lambda>L_0
		\end{gathered}
		\end{align}
		for some constant $L_0>0$. Then $f_2^\a:=\sigma_1^\a\circ \L^\a$ is a steady state of the Vlasov-Poisson system.
		\itemc We set $m:=\min\{m_-,m_+\}$, $\bar m:=\max\{m_-,m_+\}$. Suppose that
		$$\beta\ge \frac{2c\bar m}{e} $$ 
		and let $\sigma_3^\a \in C^1([0,\infty[\times\RR),\; \a=\pm 1$ be nonnegative functions which satisfy
		\begin{gather}
		\label{COND2}
		\begin{gathered}
		\forall \eta\ge 0, p\in\RR :\;\frac{\sigma_3^-\big(\eta\, m_-^{-1}, p\big)}{m_-^2} = \frac{\sigma_3^+\big(\eta\, m_+^{-1}, p\big)}{m_+^2} \quad\tand\quad
		\sigma_3^\a(\eta,p) = 0, \;\;\a=\pm 1, \;\;\text{if}\;\; \eta>E_0 \;\text{or}\;p>P_0
		\end{gathered}
		\end{gather}
		for some constants $E_0,P_0>0$. Then $f_3^\a:=\sigma_3^\a\circ (\E,\P^\a)$ is a steady state of the Vlasov-Poisson system that is compactly supported with radius $\sqrt{4P_0+2 E_0/m}$ and has finite total charge. In particular,
		\begin{align*}
		f_3^\a(x,v) = 0\quad\text{if}\quad |v|>\sqrt{2E_0/m} \quad\text{or}\quad |x|>\sqrt{4P_0+2E_0/m}.
		\end{align*}
		
	\end{itemize}
\end{theorem}

\begin{proof}
	It follows from the chain rule that $f^\a_1$ is also a solution of \eqref{SVEV}. Its corresponding charge density, that is given by \eqref{SREV}, satisfies
	\begin{align*}
	\rho_1^-(x) \equiv \int \sigma_1^-\left(\tfrac 1 2 m_- |v|^2\right)\dv = \int \sigma_1^+\left(\tfrac 1 2 m_+ |v|^2\right)\dv \equiv\rho_1^+(x), \quad x\in\RR^2,
	\end{align*}
	and thus \eqref{SPEV} is verified. Hence, $f^\a_1$ is a steady state of the Vlasov-Poisson system which proves (a). Of course, $f^\a_2$ also satisfies \eqref{SVEV}. For all $x\in\RR^2$ we obtain by the changes of variables $w=\a\tfrac{cm_\a}{e\beta}\,v^\bot$ that
	\begin{align*}
	\rho_2^\a(x) &= \int \sigma_2^\a\big( \L^\a(x,v) \big) \dv 
	= \int \sigma_2^\a\left( \tfrac 1 2\big|x+\a\tfrac{cm_\a}{e\beta}v^\bot\big|^2\right) \dv 
	= \tfrac{\beta^2 e^2}{c^2}\int {m_\a^{-2}}\; \sigma_2^\a\left(\tfrac 1 2 |x+w|^2 \right) \mathrm dw\\
	& = \tfrac{\beta^2 e^2}{c^2}\int {m_\a^{-2}}\; \sigma_2^\a\left( \tfrac 1 2 |w|^2 \right) \mathrm dw
	= \tfrac{\beta^2 e^2}{c^2}\int\limits_0^{L_0} \frac{\sigma_2^\a (\lambda )}{m_\a^2\sqrt{2\lambda}}\; \mathrm d\lambda.
	\end{align*}
	This means that $\rho^\a$ is constant and it follows from condition \eqref{COND1} that $\rho^+=\rho^-$ and thus $E= 0$.  This proves item (b).  Obviously, $f^\a_3$ is also a solution of the stationary Vlasov equation \eqref{SVE} due to product rule and chain rule. Now, substituting $w=\a\tfrac{cm_\a}{e\beta}\,v^\bot$ yields
	\begin{align*}
	\rho_3^\a(x) 
	&= \int \sigma_3^\a\big(\E(x,v),\P^\a(x,v)\big) \dv 
	= \int \sigma_3^\a\left(\tfrac 1 2 m_\a |v|^2 \comma \tfrac 1 2\big|x+\a\tfrac{cm_\a}{e\beta}v^\bot\big|^2 - \tfrac 1 2\big(\tfrac{cm_\a}{e\beta}|v|\big)^2 \right) \dv \\
	&= \tfrac{\beta^2 e^2}{c^2}\int {m_\a^{-2}}\; \sigma_3^\a\Big( \tfrac{\beta^2 e^2}{2c^2} m_\a^{-1}|w|^2 \comma \tfrac 1 2 |x+w|^2 - \tfrac 1 2|w|^2 \Big) \mathrm dw
	\end{align*}
	and consequently $\rho^+=\rho^-$ due to condition \eqref{COND2}. Hence, $f_3^\a$ is a steady state of the Vlasov-Poisson system. As $\sigma_3^\a(\eta,p)=0$ for $\eta>E_0$, it directly follows that $f_3^\a(x,v)=0$ if $|v|>\sqrt{2E_0/m}$. Now, since $\beta\ge 2c\bar m/e,$ we can conclude that for all $(x,v)\in\RR^4$ with $|v|\le \sqrt{2E_0/m}$ and $|x|>|v|$,
	\begin{align*}
	\P^\a(x,v) \;\ge\; \tfrac 1 2 |x|^2 - \tfrac{m_\a c}{e\beta}|x||v| \;\ge\; \tfrac 1 2 |x|^2 - \tfrac 1 2 |x||v| \;\ge\; \tfrac 1 2 |x|^2 - \tfrac 1 4 |x|^2 - \tfrac 1 4 |v|^2 \;\ge\; \tfrac 1 4 |x|^2 -\tfrac 1 {2m} E_0 .
	\end{align*}
	Thus, we obtain that $\P^\a(x,v)>P_0$ for all $(x,v)\in\RR^4$ with $|v|\le \sqrt{2E_0/m}$ and $|x|>\sqrt{4P_0+2E_0/m}$. Hence, by condition \eqref{COND2}, 
	\begin{align*}
	f_3^\a(x,v) = 0\quad\text{if}\quad |v|>\sqrt{2E_0/m} \quad\text{or}\quad |x|>\sqrt{4P_0+2E_0/m}.
	\end{align*}
	Now, since $f^\a_3$ is continuous and has compact support we can easily conclude that its charge $\|f_3^\a\|_{L^1(\RR^4)}$ is finite. This proves item (c).
\end{proof}

Note that for any given radius $R>0$, the constants $E_0$ and $P_0$ can be adjusted in such a way that $\sqrt{4P_0+2E_0/m}<R$. This means that, if the radius $R$ is chosen smaller than the radius of the reaction chamber, the plasma described by the steady state is confined in the cylindrical reactor as it keeps a certain distance to the boundary. 

\section{Steady states with a nontrivial electric potential}

Now, we want to investigate steady states in the general case where the self-consistent electric potential does not vanish. This is a lot more difficult as we have to ensure that Poisson's equation (which does not degenerate in this case) is also satisfied. Moreover, we observe that the microscopic kinetic energy $\tfrac 1 2 m_\a |v|^2$ is no more a conserved quantity of the characteristic flow as the potential energy also has to be taken into account. But even the total energy $\tfrac 1 2 m_\a |v|^2 + \a e \phi(x) $ is, in general, not constant along solutions of the characteristic system and thus it does not satisfy the stationary Vlasov equation. 
However, it actually \textit{is} a conserved quantity of the characteristic flow if the potential $\phi$ depends only on $|x|$. Therefore, we are looking for solutions of the stationary Vlasov-Poisson system (\ref{SVE},\ref{SPE},\ref{SRE}) that are cylindrically symmetric. These are solutions that are invariant under simultaneous rotations of both $x$ and $v$, i.e.
\begin{align*}
f^\a(x,v) = f^\a(Ax,Av) \quad\text{for all}\; A\in SO(2) \;\text{and}\; \a=\pm 1.
\end{align*}
We can conclude that a cylindrically symmetric solution depends only on ${r=|x|}$, $u=|v|$ and the angle $\theta$ between $x$ and $v$. This means that the solution $f^\a:[0,\infty[\times\RR^4\to[0,\infty[$ can be expressed by a function 
\begin{align*}
\sf^\a:[0,\infty[^2\times[0,2\pi[\;\to[0,\infty[ \twith \sf^\a(r,u,\theta)=f^\a(x,v) .
\end{align*}
In general we will use the non-italic version of a (roman or greek) letter to denote the equivalent description of a cylindrically symmetric function by the variables $r$, $u$ or $\theta$. We will now show that then the quantities $\rho^\a$ and $\phi$ depend only on $r$. First we can express $\rho^\a$ by polar coordinates:
\begin{align}
\rho^\a(x)=\int\limits_{\RR^2} f^\a(x,v)\dv = \int\limits_0^\infty \int\limits_0^{2\pi} \sf^\a(r,u,\theta)\; u\; \mathrm d\theta\mathrm du =:\uprho^\a(r).
\end{align}
If the solution $f^\a$ is compactly supported, so is $\rho^\a$. In this case, there exists some radius $R>0$ that does not depend on $\a$ such that $\uprho^\a(r)=0$ for $r>R$. 
Recall that the two-dimensional Laplace operator for functions that depend only on $r$ can be expressed by $\laplace = r^{-1} \ddr \left( r\ddr \right)$. Hence, we are now looking for a solution $r\mapsto \upphi(r)$ of the ordinary differential equation
\begin{align}
\label{ODE:PHI}
\ddr \left(r\ddr \upphi(r) \right) = -4\pi e \, r \uprho(r), \quad r \ge 0
\end{align}
where $\uprho:=\uprho^+-\uprho^-$. Because of symmetry, it must hold that $\upphi'(0)=0$. We will also assume that $\upphi(0)=0$ and omit the boundary condition at infinity instead. Applying the fundamental theorem of calculus twice on \eqref{ODE:PHI} we obtain that  $\upphi:[0,\infty[\to\RR$ must satisfy the integral equation
\begin{align}
\upphi(0)=0,\qquad \upphi(r)= -4\pi e \int\limits_0^r \frac 1 s \int\limits_0^s \tau\uprho(\tau)\dtau\ds, \;\; r>0.
\end{align}
Note that $\upphi$ is then continuously differentiable on $[0,\infty[$ with\vspace{-2mm}
\begin{align}
\upphi'(0)=0,\qquad \upphi'(r)= -\frac{4\pi e}{r} \int\limits_0^r \tau\uprho(\tau)\dtau, \;\; r>0.
\end{align}
Moreover, the function $r\mapsto r\upphi'(r)$ is also continuously differentiable on $[0,\infty[$ and \eqref{ODE:PHI} is satisfied in the classical sense. But how does this definition of $\upphi$ fit together with the definition of the potential $\phi$ in Cartesian coordinates? First recall that both functions $x\mapsto \phi(x)$ and $x\mapsto \upphi(|x|)$ are a solution of Poisson's equation 
\begin{align*}
- \laplace_x \phi = 4\pi e \suma \a \rho^\a
\end{align*}
by construction. This means that they can differ only by a harmonic function $h$, i.e.,$\phi(x)=\upphi(|x|)+h(x)$, $x\in\RR^2$. Since $\grad_x \phi(x)\to 0$ as $|x|\to\infty$ and \vspace{-2mm}
\begin{align*}
|\grad_x \upphi(|x|)| = |\upphi'(|x|)| \le \frac{4\pi e}{|x|} \int\limits_0^R \tau |\uprho(\tau)| \dtau \to 0, \quad |x|\to\infty
\end{align*}
we can conclude that $\grad_x h(x)\to 0 $ as $|x|\to\infty$. As the partial derivatives $\delxi h$, $i=1,2$ are still harmonic it follows that $\grad_x h = 0$ and thus $h$ is constant. Hence, $x\mapsto \phi(x)$ and $x\mapsto \upphi(|x|)$ differ only by a constant and therefore they induce the same electric field that can be expressed by
\begin{align*}
E(x) = -\grad_x \phi(x) \quad\text{or}\quad E(x) = \begin{cases} \;-\frac{x}{|x|} \upphi'(|x|) &\text{if}\; x\neq 0 \\ \; 0 &\text{if}\; x=0 \end{cases}
\end{align*}
respectively. In the following we will interpret the expression $\frac{x}{|x|} \upphi'(|x|)$ at the point $x=0$ as zero.\\[1ex]

Using these results, the stationary Vlasov-Poisson system (\ref{SVE},\ref{SPE},\ref{SRE}) can be transformed into cylindrical coordinates. The transformed system reads as follows:
\begin{empheq}[left=\empheqlbrace]{align}
&\;\begin{aligned}
\label{CSVE}
&\;u\cos\theta\, \partial_r\sf^\a(r,u,\theta)
- \frac{\a e}{m_\a} \cos\theta\,\upphi'(r)\, \partial_u\sf^\a(r,u,\theta)\\
&\hspace{20pt} + \left[ \frac{\a e}{m_\a}\left( \frac{\sin\theta}{u}\upphi'(r) - \frac{\beta}{c} \right) 
- \frac u r \sin\theta\right] \partial_\theta\sf^\a(r,u,\theta)
= 0, \quad \a=\pm 1,
\end{aligned}
\\[1ex]
\label{CSPE}
&\;\upphi'(0)=0,\qquad \upphi'(r):= -\frac{4\pi e}{r} \suma \a \int\limits_0^r \tau\uprho^\a(\tau)\dtau, \;\; r>0,
\\[1ex]
\label{CSRE}
&\;\uprho^\a(r) = \int \sf^\a(r,u,\theta)\, u\;\mathrm d\theta\mathrm du, \quad \a=\pm 1.
\end{empheq}
Note that, if $f^\a(x,v)=\sf^\a(r,u,\theta)$, equation \eqref{CSVE} is equivalent to the Cartesian formulation
\begin{align}
\label{CSVE2}
&\;v\cdot \partial_x f^\a + \frac{\a e}{m_\a}\left[-\frac{x}{|x|}\upphi'(|x|) + \frac{b}{c}v^\bot \right]\cdot \partial_v f^\alpha = 0, \quad \a=\pm 1 .
\end{align}
which is sometimes more convenient. \\[1ex]
Now, the goal is to find a cylindrically symmetric solution $f^\a(x,v)=\sf^\a(r,u,\theta)$ of this system (\ref{CSVE},\ref{CSPE},\ref{CSRE}). 
The basic strategy to construct such a steady state is the following: We suppose that $\upphi=\upphi(r)$ is a given cylindrically symmetric electric potential and that the magnetic field is given by a constant $b\equiv \beta >0 $. Then the quantities 
\begin{align*}
\E^\a(x,v)&\;:=\; \tfrac 1 2 m_\a |v|^2 + {\a e}\, \upphi(|x|) &&\hspace{-32mm}= \tfrac 1 2 m_\a u^2 + \a e\, \upphi(r) 
&&\hspace{-32mm} =: \text{E}^\a(r,u,\theta),\\[2mm]
\P^\a(x,v)&\;:=\; \tfrac 1 2 |x|^2 + \alpha \tfrac{cm_\a}{e\beta}x\cdot v^\bot &&\hspace{-32mm}= \tfrac 1 2 r^2 + \alpha \tfrac{cm_\a}{e\beta}\, ru \sin\theta &&\hspace{-32mm} =: \text{P}^\a(r,u,\theta)
\end{align*}
are constant along solutions of the characteristic system
$$ \dot x = v, \quad \dot v = \tfrac{\a e}{m_\a}\left[ - \tfrac{x}{|x|} \upphi'(|x|)+ \tfrac \beta c v^\bot\right] $$
and hence, they are solutions of the stationary Vlasov equation \eqref{CSVE}. The quantity $\E^\a(x,v)$ describes the total energy and $m_\a x\cdot v^\bot$ denotes the angular momentum of a single particle at the point $x$ with velocity $v$. \\[1ex]
This leads to the following ansatz for a steady state:
\begin{align}
\label{ANS}
f^\a(x,v) = \sigma^\a \big(\E^\a(x,v) \comma \P^\a(x,v)\big), \quad a=\pm 1
\end{align}
where the functions $\sigma^\a \in C^1_b(\RR;[0,\infty[), \; \a=\pm 1$ are supposed to satisfy the condition
\begin{align}
\label{COND:SIGMA} 
\left\{
\begin{aligned}
&\sigma^\a(\eta,p) = 0 &&\text{if}\; \eta\ge E_0 \;\text{or}\; p\ge 0,\\
&\sigma^\a(\eta,p) > 0 &&\text{if}\; e_0<\eta<E_0\;\text{and}\; p_0<p< 0,\\
&|\sigma^\a(\eta,p)| \le \Sigma^\a(\eta) &&\text{for all}\; (\eta,p)\in\RR^2,
\end{aligned}
\right.
\end{align}
for constants $E_0>0$, $e_0<E_0$, $p_0<0$ and given functions $\Sigma^\a\in C(\RR;[0,\infty[)\cap L^1(\RR)$, $\a=\pm 1$.  By product rule and chain rule one can easily show that $f^\a$ satisfies the stationary Vlasov equation \eqref{CSVE}. However, to obtain a self-consistent steady state of the Vlasov-Poisson system, Poisson's equation \eqref{CSPE} must be satisfied as well. By this ansatz the cylindrically symmetric charge density ${\uprho^\a = \uprho^\a(r)}$ becomes a functional of the potential $\upphi$. It is given by
\begin{align*}
\uprho^\a(r) = \int\limits_0^\infty \int\limits_0^{2\pi} \sigma^\a \Big( \tfrac 1 2 m_\a u^2 + \a e\, \upphi(r) \comma
\tfrac 1 2 r^2 + \alpha \tfrac{cm_\a}{e\beta}\, ru\sin\theta\Big) \; u \;\mathrm d\theta\mathrm du 
\end{align*} 
By the change of variables $E=\frac 1 2 m_\a u^2 + \a e\, \upphi(r)$ we obtain that 
\begin{align}
\label{EQ:RHOG}
\uprho^\a(r) = g^\a\big(r,\upphi(r)\big) ,\quad r\ge 0
\end{align} 
where the function $g^\a:\RR^2\to [0,\infty[$ is defined by 
\begin{align}
\label{DEF:G}
g^\a(r,\varphi) = 
\tfrac{1}{m_\a}
\int\limits_{\a e\varphi}^{E_0}  
\int\limits_0^{2\pi} \sigma^\a\Big(E \comma \tfrac 1 2 r^2 + \tfrac{\a cm_\a}{e\beta}\, r \sqrt{\tfrac{2}{m_\a}\big(E-{\a e} \varphi\big)}\sin\theta\Big) \,\mathrm d\theta \mathrm dE 
\qquad\text{if}\;\; r\ge 0 \;\;\text{and}\;\;\a e \varphi < E_0,
\end{align} 
and $g^\a(r,\varphi)=0$ else. Now the problem reduces to finding a suitable potential $\upphi$ which satisfies the integral equation 
\begin{align}
\label{ANS:POT}
\upphi(0)=0,\qquad
\upphi(r) =-4\pi e \int\limits_0^r \frac 1 s \int\limits_0^s \tau \, g\big(\tau,\upphi(\tau)\big) \ds,\quad r>0
\end{align}
where $g:=g^+-g^-$. If the function $\upphi \in C^2([0,\infty[)$ satisfies this integral equation, the ansatz \eqref{ANS} actually \textit{defines} a steady state of the Vlasov-Poisson system whose induced electric potential is $\upphi$. We set
\begin{align*}
m:=\min\{m_-,m_+\}, \quad \bar m:=\max\{m_-,m_+\}
\end{align*}
and in the following we will assume that the exterior magnetic field is sufficiently strong, in particular,
\begin{align}
\label{COND:BETA}
\beta\ge \frac{2c\bar m}{e}.
\end{align}
One can show that the ansatz \eqref{ANS} yields a compactly supported pair of functions $f^\a$ if the chosen potential $\upphi$ satisfies a certain growth estimate:

%
%

\begin{proposition}
	\label{PROP:CSUPP}
	Suppose that $E_0,c_*\ge 0$ and $R>0$ satisfy the equation
	\begin{align}
	\label{CER}
	\tfrac{ec_*}{m} + \sqrt{\tfrac{2}{m}E_0+\left(\tfrac{ec_*}{m}\right)^2} = R,
	\end{align}
	let $\upphi\in C([0,\infty[)$ be any given potential which satisfies
	\begin{align}
	|\upphi(r)|\le c_*r,\quad r\ge 0
	\end{align} 
	and let $f^\a$ be defined by the ansatz \eqref{ANS}. Then, for both $\a=\pm 1$, $f^\a$ is compactly supported with radius $R$ and has finite charge. Furthermore, $\uprho^\a = g^\a\big(\cdot,\upphi(\cdot)\big)$ is compactly supported with
	\begin{align*}
	\uprho^\a(r) = g^\a\big(r,\upphi(r)\big) = 0 \quad \text{if}\;\; r>R.
	\end{align*} 
\end{proposition}

\bigskip

\begin{proof}
	Let $x,v\in\RR^2$ be arbitrary. First we assume that $|v|\ge|x|$. Then 
	\begin{align*}
	\E^\a(x,v) \ge \tfrac 1 2 m |v|^2 - e\big|\upphi(|x|)\big| \ge \tfrac 1 2 m |v|^2 -e c_* |v| 
	\end{align*}
	Hence $\E^\a(x,v)>E_0$ if $|v|\ge |x|$ and
	\begin{align*}
	|v|^2 - 2\tfrac{ec_*}{m} |v| + \big(\tfrac{ec_*}{m}\big)^2 > \tfrac 2 m E_0+\big(\tfrac{ec_*}{m}\big)^2 \quad\Leftrightarrow\quad |v|> \tfrac{ec_*}{m} + \sqrt{\tfrac{2}{m}E_0+\big(\tfrac{ec_*}{m}\big)^2} = R.
	\end{align*}
	From condition \eqref{COND:SIGMA} we can conclude that
	\begin{align*}
	f^\a(x,v) = 0 \quad \text{if}\;\; |v|>R \;\;\text{and}\;\; |v|\ge |x|.
	\end{align*}
	For any $(x,v)\in\RR^4$ with $|v|<|x|$, we can use \eqref{COND:BETA} to obtain that
	\begin{align}
	\label{EST:PA}
	\P^\a(x,v) \ge \tfrac 1 2 |x|^2 - \tfrac{cm_\a}{e\beta} |x||v| \ge \tfrac 1 2 |x|^2 - \tfrac 1 2 |x||v| \ge \tfrac 1 4 |x|^2 - \tfrac 1 4 |v|^2 > 0
	\end{align}
	and thus 
	\begin{align*}
	f^\a(x,v) = 0 \quad \text{if}\;\; |v|>R .
	\end{align*}
	Let us now assume that $|v|\le R$. Then \eqref{EST:PA} yields
	\begin{align}
	\P^\a(x,v) \ge \tfrac 1 4 |x|^2 - \tfrac 1 4 R^2.
	\end{align}
	Consequently $\P^\a(x,v)$ is positive if $|x|>R$ and $|v|\le R$ and by condition \eqref{COND:SIGMA} this finally implies that
	\begin{align*}
	f^\a(x,v) = 0 \quad \text{if}\;\; |v|>R \;\; \text{or}\;\; |x|>R.
	\end{align*}
	It immediately follows that $f^\a$ has finite charge and that $g^\a\big(r,\upphi(r)\big)=\uprho^\a(r)=0$ if $r>R$.
\end{proof}

\bigskip

Now, Proposition \ref{PROP:CSUPP} can be used to construct a suitable solution of the integral equation \eqref{ANS:POT} which in turn satisfies the growth condition of Proposition \ref{PROP:CSUPP}:

\begin{proposition}
	\label{PROP:PHI}
	Let $R>0$ be arbitrary and suppose that $\sigma^\a,\; \a=\pm 1$ satisfy condition \eqref{COND:SIGMA}. Let $c_0 \ge 0$ be defined by
	\begin{align*}
	c_0 := 8\pi^2 e\suma \tfrac{1}{m_\a} \|\Sigma^\a\|_{L^1(\RR)}.
	\end{align*}
	and let $E_0>0$ be the solution of 
	\begin{align}
	\label{EQ:E0}
	\tfrac{e c_0R}{m} + \sqrt{\tfrac 2 m E_0 + \big(\tfrac{e c_0R}{m}\big)^2} = R. 
	\end{align} 
	Then the integral equation \eqref{ANS:POT}, that is
	\begin{align*}
	\upphi(r) =-4\pi e \int\limits_0^r \frac 1 s \int\limits_0^s \tau \, g\big(\tau,\upphi(\tau)\big) \ds,\quad r\ge 0,
	\end{align*}
	has a unique solution $\upphi\in C([0,\infty[)$. This solution is even twice continuously differentiable and satisfies $|\upphi(r)|\le (c_0R)r$ for all $r\ge 0$. Moreover, it holds that $g\big(r,\upphi(r)\big)=0$ if $r>R$.
\end{proposition}

\bigskip

\begin{proof}
	First note that $E_0$ is well defined because of the intermediate value theorem and the strict monotonicity of the left-hand side of \eqref{EQ:E0} with respect to $E_0$. We consider a recursive sequence that is defined by
	\begin{align*}
	\upphi_0(r) := 0, \qquad \upphi_{n+1}(r) := -4\pi e\int\limits_0^r \frac 1 s \int\limits_0^s \tau g\big(\tau,\upphi_n(\tau)\big) \dtau \ds,\;\; n\in\NN \\[-8mm]
	\end{align*}
	for all $r\ge 0$.\\[1ex]
	
	\pagebreak[2]
	
	\textit{Assertion 1.} For any $n\in\NN_0$:\vspace{-1mm}
	\begin{itemize}
		\itemi $\upphi_n \in C([0,\infty[)$ and $|\upphi_n(r)| \le (c_0R)\, r$ for all $r\ge 0$. 
		\itemii $g\big(r,\upphi_n(r)\big)=0$ for all $r>R$.
	\end{itemize}
	
	\textit{Proof of Assertion 1.} For $n=0$ assertion (i) is obvious and Proposition~\ref{PROP:CSUPP} implies (ii). Now suppose that $n$-th iterate is already constructed and satisfies (i) and (ii). Then the function $r\mapsto g(r,\upphi_n(r))$ is continuous with $g(r,\upphi_n(r))=0$ for all $r>R$ and hence $\upphi_{n+1} \in C([0,\infty[)$ easily follows.  Moreover, from the definition of $g$ and condition \eqref{COND:SIGMA} we infer that
	\begin{align*}
	\big|g^\a\big(\tau,\upphi_n(\tau)\big)\big| 
	\le \frac{1}{m_\a}
	\int\limits_{\a e\varphi}^{E_0}  
	\int\limits_0^{2\pi} \Sigma^\a(E) \;\mathrm d\theta \mathrm dE 
	\le \frac{2\pi}{m_\a} \|\Sigma^\a\|_{L^1(\RR)}
	\qquad\text{if}\;\; r\ge 0 \;\;\text{and}\;\;\a e \varphi < E_0.
	\end{align*} 
	This implies that 
	\begin{align*}
	|\upphi_{n+1}(r)| \le 4\pi e \int\limits_0^r\int\limits_0^R \big|g\big(\tau,\upphi_n(\tau)\big)\big| \dtau \le (c_0 R)\, r,\quad r\ge 0.
	\end{align*} 
	Again, item (ii) follows directly from Proposition~\ref{PROP:CSUPP} with $c_*=c_0R$.  This proves Assertion 1 by induction.\\[1ex]
	
	\textit{Assertion 2.} The sequence $(\upphi_n)_{n\in\NN}$ is uniformly Cauchy on any interval $[0,\delta]$ with $\delta>0$. 
	
	\medskip
	
	\textit{Proof of Assertion 2.} First note that, for $\a=\pm 1$, 
	\begin{align*}
	g^\a(r,\varphi) = 
	\tfrac{1}{m_\a}
	\int\limits_{0}^{E_0-\a e \varphi}
	\int\limits_0^{2\pi} \sigma^\a\Big(E+\a e \varphi \comma \tfrac 1 2 r^2 + \tfrac{\a cm_\a}{e\beta}\, r \sqrt{\tfrac{2}{m_\a}E}\,\sin\theta\Big) \,\mathrm d\theta \mathrm dE 
	\qquad\text{if}\;\; r\ge 0 \;\;\text{and}\;\;\a e \varphi < E_0,
	\end{align*} 
	and $g^\a(r,\varphi) =0$ else. Hence, by Leibnitz's rule we can conclude that the function $(r,\varphi)\mapsto g^\a(r,\varphi)$ is partially differentiable with respect to $\varphi$ with 
	\begin{align*}
	\partial_\varphi g^\a(r,\varphi) &=
	\tfrac{\a e}{ m_\a} \, \int\limits_{0}^{E_0-\a e \varphi} 
	\int\limits_0^{2\pi} \partial_\eta \sigma^\a\Big(E+\a e\varphi \comma \tfrac 1 2 r^2 + \tfrac{\a cm_\a}{e\beta}\, r \sqrt{\tfrac{2}{m_\a}E}\,\sin\theta\Big)  \,\mathrm d\theta \mathrm dE 
	\qquad\text{if}\;\; r> 0 \;\;\text{and}\;\;\a e \varphi < E_0
	\end{align*} 
	and $\partial_\varphi g^\a(r,\varphi) = 0$ if $r<0$ or $\a e\varphi > E_0$. One can easily see that this partial derivative is actually continuous on $\RR^2$ and satisfies the estimate 
	\begin{align*}
	\big| \partial_\varphi g^\a(r,\varphi) \big| 
	\;\le\; \frac{2\pi e}{m_\a}(E_0+  e |\varphi|)\, \|\sigma^\a\|_{C^1_b(\RR)}.
	\end{align*} 
	Thus, for any $n\in\NN$ and $(r,\varphi)\in\RR^2$ with $r\ge 0$ and $\varphi$ between $\upphi_n(r)$ and $\upphi_{n-1}(r)$, Assertion 1(i) implies that 
	\begin{align*}
	\big| \partial_\varphi g(r,\varphi) \big| & \le \frac{2\pi e}{m_\a}(E_0+  e |\varphi|)\suma \|\sigma^\a\|_{C^1_b(\RR)} 
	\le \underbrace{\frac{2\pi e}{m_\a} \max\big\{E_0, e c_0 R\big\}\suma \|\sigma^\a\|_{C^1_b(\RR)}}_{=:L} \; (1+ r).
	\end{align*} 
	Hence, the mean value theorem implies that for all $r\in[0,\delta]$ and all $n\in\NN$,
	\begin{align*}
	\big| \upphi_{n+1}(r) - \upphi_n(r) \big| 
	& \le 4\pi e \int\limits_0^r \int\limits_0^s \big| g\big(\tau,\upphi_{n}(\tau)\big) - g\big(\tau,\upphi_{n-1}(\tau)\big) \big| \dtau \ds\\
	&\le 4\pi e L (1+\delta) \int\limits_0^r \int\limits_0^s \big| \upphi_{n}(\tau) - \upphi_{n-1}(\tau) \big| \dtau \ds.
	\end{align*}
	Since $| \upphi_{1}(r) - \upphi_0(r) | \le (c_0R)\, r$ for all $r\in [0,\delta]$ it follows by induction that
	\begin{align*}
	\big| \upphi_{n+1}(r) - \upphi_n(r) \big| \le c_0R  \frac{\big(4\pi e L (1+\delta)\big)^n}{(2n+1)!} \, r^{2n+1}, \quad r\in [0,\delta], n\in\NN.
	\end{align*}
	Without loss of generality, we assume that $4\pi e L (1+\delta) \ge 1$. Now let $\eps>0$ be arbitrary and let $m,n,N\in\NN$ with $m>n\ge N$. We obtain
	\begin{align*}
	&\big| \upphi_{m}(r) - \upphi_n(r) \big| \le c_0R  \sum_{k=n}^{m-1} \frac{\big(4\pi e L (1+\delta)\big)^k}{(2k+1)!}\, r^{2k+1} \\
	& \qquad \le c_0R  \sum_{k=N}^\infty \frac{\big(4\pi e L (1+\delta)\big)^{2k+1}}{(2k+1)!}\, r^{2k+1} 
	\le c_0R  \sum_{k=N}^\infty \frac{\big(4\pi e L \delta (1+\delta)\big)^k}{k!} 
	<\eps
	\end{align*} 
	for all $r\in[0,\delta]$ if $N$ is sufficiently large as the series is convergent. This proves Assertion 2.\\[1ex]
	
	\textit{Assertion 3.} There exists some function $\upphi \in C([0,\infty[)$ such that $\upphi_n \to \upphi$ uniformly on any compact subinterval $[0,\delta]\subset [0,\infty[$ with $\delta>0$. This function $\upphi$ has the following additional properties:
	\begin{itemize}
		\itemi $\upphi$ satisfies the integral equation \eqref{ANS:POT} on $[0,\infty[$.
		\itemii  $|\upphi(r)|\le (c_0R)r$ for all $r\ge 0$. 
		\itemiii $\upphi\in C^2([0,\infty[)$.
		\itemiv $g\big(r,\upphi(r)\big)=0$ for all $r>R$.
	\end{itemize}
	
	\textit{Proof of Assertion 3.} Assertion 2 particularly implies that for any $r\ge 0$, the sequence $(\upphi_n(r))_{n\in\NN}$ is a Cauchy sequence in $\RR$ and thus convergent.  Hence, the function $\upphi$ can be defined in a pointwise sense by
	\begin{align*}
	\upphi(r):=\underset{n\to\infty}{\lim} \upphi_n(r),\quad r\ge 0,
	\end{align*}  
	Moreover, for any $\delta>0$ we can conclude from Assertion 2 and the completeness of $C([0,\delta])$ that $\upphi_n\vert_{[0,\delta]}$ converges to $\upphi\vert_{[0,\delta]} $ uniformly on $[0,\delta]$ and therefore $\upphi\vert_{[0,\delta]} \in C([0,\delta])$. As this holds for every $\delta>0$, it follows that $\upphi\in C([0,\infty[)$. By Lebesgue's convergence theorem and the continuity of $g$, we obtain that
	\begin{align*}
	\upphi(r) = \underset{n\to\infty}{\lim} \upphi_{n+1}(r) = \underset{n\to\infty}{\lim} \left\{-4\pi e\int\limits_0^r  \frac 1 s \int\limits_0^s \tau g\big(\tau,\upphi_{n-1}(\tau)\big) \dtau \ds\right\} 
	= -4\pi e \int\limits_0^r \frac 1 s \int\limits_0^s \tau g\big(\tau,\upphi(\tau)\big) \dtau \ds,\quad r\ge 0,
	\end{align*}
	which proves (i). Moreover, $\upphi$ is bounded by 
	\begin{align*}
	|\upphi(r)| \le \underset{n\to\infty}{\lim\sup}\; |\upphi_n(r)| \le (c_0R)\, r, \quad r\ge 0
	\end{align*} 
	that is (ii). Since $r\mapsto \upphi(r)$ is continuous on $[0,\infty[$, so is $r\mapsto g(r,\upphi(r))$. Hence, $\upphi$ is continuously differentiable on $[0,\infty[$ and its derivative satisfies the integro-differential equation\vspace{-2mm}
	\begin{align*}
	\upphi'(0)=0, \qquad \upphi'(r) = -\frac{4\pi e}{r} \int\limits_0^r \tau g\big(\tau,\upphi(\tau)\big) \dtau,\quad r>0.
	\end{align*}
	Now, $\upphi'$ is even once more continuously differentiable on $[0,\infty[$ with
	\begin{align*}
	\upphi''(r) &= \frac{4\pi e}{r^2} \int\limits_0^r \tau g\big(\tau,\upphi(\tau)\big) \dtau - 4\pi e\, g\big(r,\upphi(r)\big), \quad r>0,\\
	\upphi''(0) &= \underset{r\underset{>}{\to} 0}{\lim}\; \upphi''(r) =-2\pi e\, g(0,0)
	\end{align*}
	where the second line follows from L'Hôpital's rule and the continuity of $g$ and $\upphi$.
	This proves (iii). Finally, (iv) follows directly from (ii) and Proposition \ref{PROP:CSUPP} with $c_*=c_0R$. Thus Assertion 3 is established.\\[1ex]
	
	\textit{Assertion 4.} The function $\upphi$ from Assertion 3 is the unique solution of the integral equation \eqref{ANS:POT}.\\[1ex]
	
	\textit{Proof of Assertion 4.} We assume that the integral equation \eqref{ANS:POT} has another solution $\uppsi \in C([0,\infty[)$.  From the definition of $g^\a$ and condition \eqref{COND:SIGMA} we infer that 
	\begin{align*}
	\big|g^\a\big(\tau,\uppsi(\tau)\big)\big| \le \frac{2\pi}{m_\a}\, \|\Sigma^\a\|_{L^1(\RR)} ,\quad \tau\ge 0
	\end{align*} 
	and thus 
	\begin{align*}
	|\uppsi(r)| \le 4\pi e \int\limits_0^r\int\limits_0^R \big|g\big(\tau,\uppsi(\tau)\big)\big| \dtau \le (c_0 R)r,\quad r\ge 0.
	\end{align*} 
	Now, analogously to the proof of Assertion 2, we can conclude that for any arbitrary $\delta>0$ and all $r\in[0,\delta]$,
	\begin{align*}
	\big| \upphi(r) - \uppsi(r) \big| 
	\le 4\pi e L (1+\delta) \int\limits_0^r \int\limits_0^s \big| \upphi(\tau) - \uppsi(\tau) \big| \dtau \ds 
	\le 4\pi e L \delta (1+\delta) \int\limits_0^r  \big| \upphi(\tau) - \uppsi(\tau) \big| \dtau .
	\end{align*}
	Hence, Gronwall's lemma implies that $\upphi=\uppsi$ on $[0,\delta]$ and, since $\delta$ was arbitrary, this yields $\upphi=\uppsi$ on $[0,\infty[$. This proves Assertion 4 and completes the proof of Proposition \ref{PROP:PHI}.
\end{proof}

\smallskip

Now, Proposition \ref{PROP:CSUPP} and Proposition \ref{PROP:PHI} can be used to prove the main result of this paper: 

\begin{theorem}
	\label{THM:MR}
	Let $R>0$ be arbitrary and suppose that $\sigma^\a,\; \a=\pm 1$ satisfy condition \eqref{COND:SIGMA}. Let $E_0>0$ and $c_0\ge 0$ be as defined in Proposition \ref{PROP:PHI} and let ${\upphi\in C^2([0,\infty[)}$ denote the unique solution of the integral equation \eqref{ANS:POT} on $[0,\infty[$, i.e.,
	\begin{align*}
	\upphi(r) =-4\pi e \int\limits_0^r \frac 1 s \int\limits_0^s \tau \, g\big(\tau,\upphi(\tau)\big) \ds,\quad r\ge 0 .
	\end{align*}
	with $g$ as defined above (see \eqref{DEF:G}). Then the ansatz
	\begin{align*}
	f^\a(x,v) = \sigma^\a \big(\E(x,v) \comma \P^\a(x,v)\big), \quad a=\pm 1
	\end{align*} 
	defines a steady state of the Vlasov-Poisson system that is compactly supported with radius $R$ and has finite charge. It even holds that both $f^-$ and $f^+$ have positive charge, i.e.,
	\begin{align*}
	0<\|f^\a\|_{L^1(\RR^4)} < \infty,\quad \a=\pm 1
	\end{align*}
	if $e_0$ and $p_0$ (from condition \eqref{COND:SIGMA}) are  sufficiently small real numbers. 
\end{theorem}

\noindent\textit{Comment.} To avoid misunderstandings, we remind the reader that $p_0<0$ and also $e_0$ might be negative. Therefore, in contrast to frequent usage, the phrase ``sufficiently small'' does not mean that $p_0$ and $e_0$ are close to zero.

\smallskip

\begin{proof}
	The solution $\upphi$ of the integral equation \eqref{ANS:POT} is well defined by Proposition \ref{PROP:PHI}. As already mentioned, $f^\a$ is a solution of the stationary Vlasov equation. According to \eqref{EQ:RHOG} the corresponding volume charge density is given by $\uprho^\a(r)=g^\a\big(r,\upphi(r)\big)$ and thus, by its construction, $\upphi$ is the corresponding self-consistent electric potential. This means that the pair $(f^\a)_{\a=\pm 1}$ is a steady state of the Vlasov-Poisson system. Note that Proposition \ref{PROP:PHI} also implies that 
	\begin{align*}
	|\upphi(r)| \le (c_0R) r, \quad r\ge 0.
	\end{align*} 
	Then, it follows from Proposition \ref{PROP:CSUPP} that, for both $\a=\pm 1$, $f^\a$ is compactly supported with $\supp f^\a \subset\BR\times\BR$ and thus, the steady state has finite charge. Moreover, for $\a=+1$, 
	\begin{align}
	\|f^\a\|_{L^1(\RR^4)} &= 2\pi \int\limits_0^\infty \uprho^\a(r)\, r \;\mathrm dr  
	= 2\pi\int\limits_0^\infty \int\limits_0^\infty \int\limits_0^{2\pi} \sigma^\a\big(\text{E}^\a(r,u,\theta) \comma \text{P}^\a(r,u,\theta) \big)\, ur\, \mathrm d\theta\mathrm du\mathrm dr\notag\\
	\label{EST:CHG}
	& \ge 2\pi \int\limits_{\mathcal M} \sigma^\a\big(\text{E}^\a(r,u,\theta) \comma \text{P}^\a(r,u,\theta) \big)\, ur\; \mathrm d(r,u,\theta)
	\end{align} 
	where 
	\begin{align*}
	{\mathcal M}:=\Big\{ (r,u,\theta)\in\RR^3 \;\big\vert\; r_0<r<2r_0,\;\; u_0<u<2u_0,\;\; \theta_0<\theta<\theta_0+\tfrac \pi 2 \Big\}
	\end{align*}
	with $r_0,u_0>0$ and $0<\theta_0<\tfrac{3\pi}{2}$. We choose 
	\begin{align*}
	u_0:=r_0^{1/4} \tand \theta_0:=\left\{
	\begin{aligned} 
	&\tfrac \pi 4, &&\text{if}\;\a=-1 \phantom{\big\vert}\\
	&\tfrac {5\pi}{4}, &&\text{if}\;\a=1.
	\end{aligned}
	\right.
	\end{align*} 
	Then, for both $\a=\pm 1$, $$\a\sin\theta < -\tfrac 1 2 \sqrt{2}, \quad\text{if}\quad \theta_0<\theta<\theta_0+\frac \pi 2$$ and hence, if $r_0>0$ and $p_0<0$ are small enough, 
	\begin{align}
	\label{COND:CC1}
	\textstyle
	p_0\,<\, 
	\frac 1 2 r_0^2 - 4 \frac{c\bar m}{e\beta}\, r_0^{5/4} \,<\,
	\text{P}^\a(r,u,\theta) \,<\, 2 r_0^2 - \frac 1 2\sqrt{2}\, \frac{c m}{e\beta} r_0^{5/4} 
	\,<\, 0,\qquad (r,u,\theta)\in {\mathcal M}.
	\end{align} 
	Furthermore, 
	\begin{align}
	\label{COND:CC2}
	\textstyle 
	e_0 \,<\,
	\frac 1 2 m\, r_0^{1/2} -  2ec_0Rr_0 \,<\,
	\text{E}^\a(r,u,\theta) \,<\, 2\bar m\, r_0^{1/2} + 2ec_0Rr_0 \,<\,E_0, 
	\qquad (r,u,\theta)\in {\mathcal M}.
	\end{align} 
	if $r_0>0$ and $e_0<E_0$ are sufficiently small. This means that, assuming that $p_0$ and $e_0$ are  sufficiently small real numbers, we can fix $r_0$ such that 
	\begin{align*}
	\sigma^\a\big(\text{E}^\a(r,u,\theta) \comma \text{P}^\a(r,u,\theta) \big)>0 \quad\text{for all}\quad (r,u,\theta)\in {\mathcal M}
	\end{align*} 
	which directly implies that $\|f^\a\|_{L^1(\RR^4)}>0$ because of \eqref{EST:CHG} and the fact that ${\mathcal M}$ has positive measure. This completes the proof of Theorem \ref{THM:MR}.	
\end{proof}
\bigskip

By this theorem we have established the existence of compactly supported steady states even with a \textit{nontrivial} self-consistent electric potential. If $e_0$ and $p_0$ are sufficiently small, the constructed steady state describes a true two-component plasma as each component has positive total charge. If we choose a radius $R>0$ that is smaller than the radius of the reaction chamber, the plasma that is described by such a steady state is confined in the cylindrical reactor since it keeps a certain distance to the reactor wall.

\section*{Acknowledgement} 
Special thanks to Gerhard Rein for helpful discussion. This work does not have any conflicts of interest.

\scriptsize
\bibliography{mmas_ssvp}

\end{document}